\documentclass{ws-ijmpc}
\usepackage[super]{cite} 
\usepackage[utf8]{inputenc}

\usepackage{amsmath,graphicx,mathabx,bbm, color}
\usepackage[normalem]{ulem}

\newcommand{\ZZ}{{\mathbbm{Z}}}
\newcommand{\NN}{{\mathbbm{N}}}

\newcommand{\B}{{\mathcal B}}
\newcommand{\A}{{\mathcal A}}
\newcommand{\calS}{{\mathcal S}}
\newcommand{\f}{{\mathbf f}}

\begin{document}

\markboth{Henryk Fuk\'s \& Yucen Jin}
{Approximating dynamics of a number-conserving cellular automaton by
a finite-dimensional dynamical system}

\catchline{}{}{}{}{}

\title{Approximating dynamics of a number-conserving cellular automaton by
a finite-dimensional dynamical system
}

\author{Henryk Fuk\'s}

\address{Department of Mathematics and Statistics\\
         Brock University\\
         St. Catharines, ON L2S3A1, Canada\\
hfuks@@brocku.ca}

\author{Yucen Jin}

\address{Department of Mathematics and Statistics\\
         Brock University\\
         St. Catharines, ON L2S3A1, Canada\\
yj12qj@brocku.ca}

\maketitle

\begin{history}
\received{March 10, 2020}
\revised{August 24, 2020}
\end{history}

\begin{abstract}
The local structure theory for cellular automata (CA) can be viewed as an finite-dimensional approximation of infinitely-dimensional system.  While it is well known that this approximation works surprisingly well for some cellular automata, it is still not clear
why it is the case, and which CA rules have this property. In order to shed some light on this problem, we present 
an example of a four input CA for which probabilities of occurrence of short blocks of symbols can be computed
exactly. This rule is number conserving and possesses a blocking word. Its local structure approximation correctly predicts
steady-state probabilities of small length blocks, and we present a rigorous proof of this fact, without resorting
to numerical simulations. We conjecture that the number-conserving property together with the existence of
the blocking word are responsible for the observed perfect agreement between the finite-dimensional approximation
and the actual infinite-dimensional dynamical system.

\end{abstract}

\ccode{PACS Nos.: 89.75.-k, 47.11.Qr}

\section{Introduction}

The idea of mean-field approximation or mean-field theory is a well established concept
in statistical physics and related fields. In the context of lattice gas models, 
the mean field theory approximates dynamics of the infinitely-dimensional lattice gas system by 
neglecting correlations between lattice sites. 

In 1970's and 1980's, various generalizations of the mean field theory have been proposed, most notably
in works of H.J. Brascamp \cite{Brascamp71} as well as M. Fannes and A. Verbeure \cite{Fannes84}. In late 1980's, H. Gutowitz et al. applied these ideas to cellular automata (CA), proposing
the so-called local structure theory, \cite{gutowitz87a} which included mean field theory as a special case. In spite of being over three decades old, the local structure theory is still not fully understood, and many of its aspects  remain unexplored. In particular, it is still not clear why some CA
are well approximated by the local structure theory, and
how to identify such rules in large rule spaces. This problem
will be further  referred to as the ``performance problem'' of the local structure theory.

In what follows, we will demonstrate an example of a CA rule
which can be viewed as interacting particle system conserving the number of particles and which possesses an equilibrium
state exactly as predicted by the local structure theory. The number of known CA rules of this type is so far very small, and we hope that the example presented here eventually helps to shed some light on the ``performance problem''  of the local structure theory.

Dynamics of one-dimensional cellular automata (CA)  is often studied by treating them as maps in the space of probability measures
over bi-infinite strings (to be called \emph{configurations}). The meaning of this is easy to explain in simple terms. We
consider a large set of configurations drawn from a known probability distribution  (usually the Bernoulli distribution). We then apply  a given  cellular automaton rule to all these configurations. 
As a result, we obtain an assembly of configurations which (usually) is no longer distributed according to  the Bernoulli distribution, but according to some other distribution.  The cellular automaton rule, therefore, 
 transforms the initial probability distribution (or more formally, the initial probability measure)  into some other probability measure. By applying the local rule again and again, one obtains an infinite sequence of measures, to be called the \emph{orbit} of the initial measure.

Such orbits are not easy to describe and study, as the maps generating them are infinitely-dimensional.
One can, however, approximate these maps by finite-dimensional ones, and this is the basis of the aforementioned
local structure theory developed by  H. Gutowitz et al. ~\cite{gutowitz87a}.

The local structure theory has been widely used in CA research, although a relatively few rigorous results are known
about the theory. Often it is used in a following way: one constructs a 
finite-dimensional map or a system of recurrence equations following the algorithm given by Gutowitz  \cite{gutowitz87a}, and 
numerically studies the orbit of this system. Comparison of this orbit with results of direct numerical
simulations of the  CA in question often reveals an excellent agreement between the two.

The problem of comparing numerically computed orbit of local structure theory with results of numerical simulations
is that none of the two are exact. Fortunately, in recent years some techniques have been developed which allow to compute 
elements of orbits Bernoulli measures exactly \cite{paper62}, making  a more rigorous approach possible. 

In Ref. \citen{paper64}, an example of a CA rule is given for which the local structure approximation  correctly reproduces not only 
limiting values of probabilities of short block, but also the type of convergence toward
the fixed point (as a power law). The rule used in this work, namely elementary CA rule 14, possesses so-called additive invariant of
the second order \cite{Hattori91}, conserving the number of pairs 01 between consecutive iterations. One could wonder, therefore,  if
the existence of the additive invariant somewhat ``makes'' the local structure approximation to perform well.

In order to further investigate this problem, we searched for a rule with somewhat simpler additive invariant
(of the first order), which could be studied in detail. Binary rules which possesses first-order additive invariant are called
number-conserving cellular automata rules (NCCA). Among elementary CA, there is only one non trivial NCCA, namely rule
184 (rule 226, which is obtained from rule 184  by spatial reflection, has the same dynamics). This rule has been extensively
studied, and much is known about its dynamics \cite{Krug88,paper4,paper11,Belitsky2005,paper62}.

When one increases the neighbourhood size to 4 sites (e.g., one neighbour on the left and two on the right), the number
of NCCA increases to 22, and one of the most interesting ones of them is rule 56528. Its local function is given by
\begin{align} \label{defrule}
f(0000)&=
f(0001)=
f(0010)=
f(0011)=
f(0101)=
f(1000)=
f(1001)=
f(1101)=0, \nonumber \\ 
f(0100)&=
f(0110)=
f(0111)=
f(1010)=
f(1011)=
f(1100)=
f(1110)=
f(1111)=1.  
\end{align}
Since this rule conserves the number of 1s, one can interpret it as a particle system, where 1s represent individual particles, and
0s represent empty spaces. In this representation, one can show \cite{paper8} that the motion of particles will schematically
be governed by the following rules,
\begin{displaymath}
\overset{\curvearrowright}{10}1,\quad
\overset{\rotatebox{180}{$\circlearrowright$}}{1}00,\quad
\overset{\rotatebox{180}{$\circlearrowright$}}{1}1.
\end{displaymath}
This means that only a particle which has a single zero on the right will move to the right (symbol $\curvearrowright{}$),
while particles followed by 1 or by two or more zeros will stay in the same place (symbol $\overset{\rotatebox{180}{$\circlearrowright$}}{}$).

As we will see in the next section, it is possible to obtain exact expressions for probabilities of blocks of symbols of length up to 3 (and some longer ones) for this rule.

\section{Basic definitions}
Let $\A=\{0,1\}$ be called \emph{a symbol set} or \emph{alphabet}, and let 
$\calS =\{0,1\}^{\ZZ}$
 be the set of all bisequences over $\A$,  to be
called  \emph{a configuration space}. 

{\em A block} or \emph{word} of length $n$ is an ordered set $b_{0} b_{1}
\ldots b_{n-1}$, where $n\in \NN$, $b_i \in \A$.
Let $n\in \NN$ and let
$\B_n$ denote the set of all blocks of length $n$ over $\A$ and $\B$ be
the set of all finite blocks over $\A$.

For $r_l, r_r \in \NN$, a mapping $f:\{0,1\}^{r_l+r_r+1}\mapsto\{0,1\}$ will be called {\em a cellular
 automaton rule of left radius~$r_l$ and right radius $r_r$}. Alternatively, the function $f$ can be
 considered as a mapping of $\B_{r_l+r_r+1}$ into $\B_0=\A=\{0,1\}$. 

Corresponding to $f$ (also called {\em a local mapping}) we define a
 {\em global mapping}  $F:\calS \to \calS$ such that
$
(F(s))_i=f(s_{i-r_l},\ldots,s_i,\ldots,s_{i+r_r})
$
 for any $s\in \calS$.

A {\em block evolution operator} corresponding to $f$ is a mapping
 $\f:\B \mapsto \B$ defined as follows. 
Let $r_l, r_r\in \NN$ be, respectively, the left and the right radius of $f$, and let  $a=a_0a_1 \ldots a_{n-1}\in \B_{n}$
where $n \geq r_l+r_r+1 >0$. Then 
\begin{equation}
\f(a) = \{ f(a_i,a_{i+1},\ldots,a_{i+2r})\}_{i=0}^{n-r_l-r_r-1}.
\end{equation}
 Note that if
$b \in B_{r_l+r_r+1}$ then $f(b)=\f(b)$.
The set of $n$-step preimages of the block $b$ under the rule $f$
is defined as the set $\f^{-n}(b)=\{ c\in \B: \f^n(c)=b\}$. 
The notion of block preimages has been studied
in many earlier works, although in a
different context.\cite{Jen88,Jen89,Voorhees96,mcintosh2009}

Note that the block evolution operator $\mathbf{f}$ returns a block shorther than the argument by $r_l+r_r$.
For example, for the rule defined in eq.~(\ref{defrule}), we have $\mathbf{f}(001101)=010$ because 
$f(0011)=0$, $f(0110)=1$, and $f(1101)=0$. Moreover, the inverse of $\mathbf{f}$ is usually not
single-valued, for example, $\mathbf{f}^{-1}(010)=
\{001000,
001001,
001101,
010101,
110101\}$.

In this paper we will consider only the binary rule with the local function defined by eq. (\ref{defrule}), with $r_l=1$, $r_r=2$.
Binary rules are usually identified
by their Wolfram number $W(f)$  \cite{Wolfram94}. In our case, for the four-input rule  defined in eq. (\ref{defrule}),
the Wolfram number is 
\begin{equation} \label{code3}
W(f)=\sum_{x_1,x_2,x_3, x_4=0}^{1}f(x_1,x_2,x_3, x_4)2^{(2^3x_1+2^2x_2+2^1x_3+2^0x_4)}=56528.
\end{equation}

As already mentioned, a classical problem in cellular automata theory is to compute the probability of the occurrence of a given binary string 
$\mathbf{a}$
in a configuration obtained  after $n$ iterations of the rule, assuming that the initial configuration
is drawn from the Bernoulli distribution. Such probability will be denoted by $P_n(\mathbf{a})$ and called
 \emph{block probability}.
It is easy to show that if the initial distribution is Bernoulli, then 
the probability of occurrence of $\mathbf{a}$ is independent of its position
in the configuration. We will call such block probabilities \emph{shift invariant}. 

Now, let us suppose that the the probability of occurrence of 1 in the initial configuration is $p \in [0,1]$ and
the probability of occurrence of 0 is $q=1-p$. In such a case one can show that  
 the probability of the occurrence of a given binary string 
$\mathbf{a}$
in a configuration obtained  after $n$ iterations of the rule $f$ is given by
\begin{equation} \label{bernoulliP}
P_n(\mathbf{a})=\sum_{\mathbf{b}\in \f^{-n}(\mathbf{a})}  
p^{\#_1(\mathbf{b})}q^{\#_0(\mathbf{b})}.
\end{equation}
where $\#_s(\mathbf{a})$ denotes number of symbols $s$
in $\mathbf{a}$.

We will use the above results to compute block probabilities of some blocks for rule 56528. Before we proceed, let us make
one additional remark about block probabilities.  Block probabilities must satisfy
so-called \emph{Kolmogorov consistency conditions}, so that for any block $\mathbf{a}\in \B$ one has $P_n(\mathbf{a}0)+P_n(\mathbf{a}1)=P_n(\mathbf{a})$. 
For example, we must have $P_n(1)+P_n(0)=1$, $P_n(01)+P_n(00)=P_n(0)$, etc. Consistency
conditions can be used to express some block probabilities by others. One can show that
for binary strings, among probabilities of blocks of length $k$, only $2^{k-1}$ are independent \cite{paper50},
in the sense that one can choose $2^{k-1}$ block probabilities which are not linked to each other via
consistency conditions. For example, for blocks of length up to $3$, there are 14 block probabilities,
$P_n(0)$,  $P_n(1)$, $P_n(00)$, $P_n(01)$, $P_n(10)$, $P_n(11)$
$P_n(000)$, $P_n(001)$, $P_n(010)$, $P_n(011)$, $P_n(100)$, $P_n(101)$, $P_n(110)$, and $P_n(111)$.
Among them only $2^{3-1}=4$ are independent.There is some freedom in choosing which ones are
to be treated as independent, but a common choice is to take  $P_n(0)$, $P_n(00)$, $P_n(000)$, and
$P_n(010)$ as independent blocks. This is called the \emph{short block representation} (see  Ref.~\citen{paper50} for the details of the algorithm for choosing 
independent blocks). 
Using consistency conditions, one can now express the remaining blocks of length up to 3 
in terms of the aforementioned four block probabilities, as follows: 
\begin{align} \label{dependentProbs}
P_n(1)&=1-P_n(0), \nonumber \\ 
P_n(01)&=P_n(0)-P_n(00), \nonumber \\
P_n(10)&=P_n(0)-P_n(00), \nonumber \\
P_n(11)&=1-2\, P_n(0)+P_n(00), \nonumber \\
P_n(001)&=P_n(00)-P_n(000), \nonumber \\
P_n(011)&=P_n(0)-P_n(00)-P_n(010), \nonumber \\
P_n(100)&=P_n(00)-P_n(000), \nonumber \\
P_n(101)&=P_n(0)-2\, P_n(00)+P_n(000), \nonumber \\
P_n(110)&=P_n(0)-P_n(00)-P_n(010), \nonumber \\
P_n(111)&=1-3\, P_n(0)+2\, P_n(00)+P_n(010).
\end{align}
\section{Exact results: preimage sets}
We will now compute block probabilities of length up to 3 (and even beyond) using eq.~(\ref{bernoulliP}). The first thing we need
to do is to describe the structure of preimage sets $\f^{-n}(\mathbf{a})$ for some selected short blocks $\mathbf{a}$,
namely for $100$, $101$ and $010$. We will see why these three are important in the next section.
\begin{proposition}\label{prop100}
$\mathbf{f}^{-n}(100)$ has the form $$\underbrace{*...*}_{n}100\underbrace{*...*}_{2n}$$ and $\mathbf{f}^{-n}(00100)$ has the form $$\underbrace{*...*}_{n}00100\underbrace{*...*}_{2n},$$ where $*$ is an arbitrary element in $\{0,1\}$.
\end{proposition}

Proof of the first part of the above can be done by induction.
Taking $n=1$,  we notice,
by direct verification, that preimages of $100$ are
$$
\{010000,
010001,
010010,
010011,
110000,
110001,
110010, 
110011\}=\{*100**\}.
$$
Thus the proposition is indeed valid for $n=1$.

For the induction step, assume that the expression for $\mathbf{f}^{-n}(100)$
is valid for a given $n$. This means that
$$\mathbf{f}^{-(n+1)}(100)=
\mathbf{f}^{-1} \left(\mathbf{f}^{-n}(100)\right)
=\mathbf{f}^{-1} \left(
\underbrace{*...*}_{n}100\underbrace{*...*}_{2n}
\right).
$$
Becasue the preimage of 
$100$ is $*100**$, and because
$\mathbf{f}^{-(n+1)}(100)$
must be longer that
$\mathbf{f}^{-n}(100)$ by three
symbols, we conclude that
$\mathbf{f}^{-(n+1)}(100)$ has the form
$$\underbrace{*...*}_{n+1}100\underbrace{*...*}_{2n+2}.$$
This verifies the induction step, proving the first part of the proposition.
Proof of the second part is similar.

Note that Proposition 1 implies that every block 100 stays in the same place during iterations of the rule.
This means that no information can pass through the block 100, neither from the left of from the right. We call
such a block the \emph{blocking word} \cite{Kurka1997}.

\begin{proposition} \label{prop101}
The set of $n$-step preimages of 101 under the rule 56528
is given by
$$\mathbf{f}^{-n}(101)= \bigcup_{i=0}^{n} A_{n,i},$$ where each $A_{n,i}$ is the set of all binary strings of length $3+3n$ of the form
$$\underbrace{*...*}_{i}1a_{i+2}...a_{2n}101\underbrace{*...*}_{n},$$
such that the block $a_{i+2}...a_{2n}$ has exactly $n-i$ zeros  and that it does not include any 00.
\end{proposition}

\begin{proof}
In order to to avoid tedious details we will prove the above proposition in somewhat informal way,
although every step of the following reasoning  could easily be formalized.

Figure~\ref{pat1} shows an example of a spatiotemporal pattern produced by rule 56528. One can think of the dynamics of this rule as ``movement'' of zeros in the background of ones. Isolated zeros move to the left
one cell per time step, while clusters of two or more zeros keep their left boundary in place.
When the isolated zero collides with the cluster of zeros, the cluster ``absorbs'' the isolated zero
and extends its right boundary by one (that is, it grows by one unit to the right).

As a consequence of this, the only way to obtain 101 (or isolated zero)
at time step $n+1$ is to have it at time step $n$ located at the position one unit to the right compared to step $n$ (recall that isolated zeros travel to the left), and to make sure that this zero does not get
absorbed by the nearest cluster of zeros on the left.

By induction, the only way to obtain 101 (or isolated zero)
after $n$ iterations is to have 101 in the initial string located at at the position $n$ unit to the right compared to its position after $n$ iterations, and preceded
by sufficiently long ``buffer'' which does not contain double zeros.
What is on the right of $101$ in the initial string does not matter,
and what precedes the buffer does not matter either, providing that it is sufficiently long.
This means that $\mathbf{f}^{-n}(101)$ must be of the form 
$$\underbrace{*...*}_{i}a_{i+1}...a_{2n}101\underbrace{*...*}_{n},$$
where $a_{i+1}...a_{2n}$ is the aforementioned ``buffer'' containing no double zeros. 
This buffer has length $2n-i$, where $i$ can vary from $0$ to $n$. 

Suppose now that the buffer has only ones, no zeros. Its length can then be just $n$, as show in the example in Figure \ref{stringiterates}a for $n=3$.
Set of all strings  with only ones in the buffer will be, therefore, of the form
$$\underbrace{*...*}_{n}a_{n+1} a_{n+2}...a_{2n}101\underbrace{*...*}_{n},$$
where all symbols $a_i$  for $i=n+1,\ldots 2n$ take value 1. We will
call this set $A_{n,n}$.

If the buffer has exactly one zero, it must be by one 
unit longer than before, such as examples in 
Figure \ref{stringiterates}b or \ref{stringiterates}c.
Set of all strings  with single 0  in the buffer will be, therefore, of the form
$$\underbrace{*...*}_{n-1}a_{n} a_{n+1}...a_{2n}101\underbrace{*...*}_{n},$$
where $a_{n} a_{n+1}...a_{2n}$ includes  only one zero and starts with $a_{n}=1$. We will
call this set $A_{n,n-1}$.

This pattern of construction of sets $A_{n,i}$   continues with decreasing $i$, each consecutive $A_{n,i}$ containing preimages with buffer with exactly $n-i$ zeros and $a_{i+1}=1$. The last one, $A_{n,0}$, will be the set of strings with the buffer containing exactly $n$ zeros,
such as the example in Figure \ref{stringiterates}d.

Once can easily conclude, therefore, that the set of preimages of 101 will be  the union of sets $A_{n,i}$,
each containing strings of the form
$$\underbrace{*...*}_{i}a_{i+1}a_{i+2}...a_{2n}101\underbrace{*...*}_{n},$$
such that the block $a_{i+1}...a_{2n}$ has exactly $n-i$ zeros, starts with $a_{i+1}=1$,   and  does not include any 00, exactly  as claimed. 
\end{proof}

\begin{figure}
\begin{center}
\includegraphics[width=12cm]{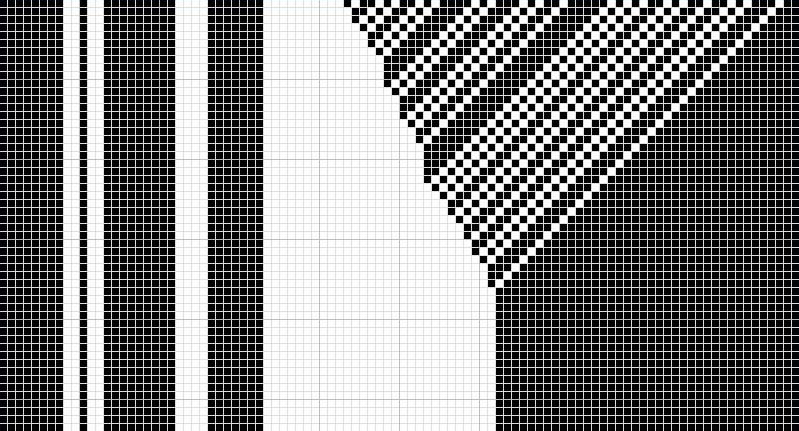} 
\end{center}
\caption{Spatiotemporal pattern generated by rule 56528, using lattice of 100 sites with
periodic boundaries. Black squares represent 1s and white squares represent 0s. Time (consecutive iterations) proceeds downwards.}\label{pat1} 
\end{figure} 

\begin{figure}
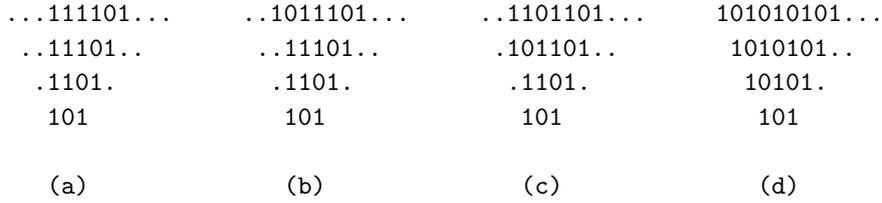

\begin{center}
\begin{minipage}{3cm}
\begin{verbatim}
 ...111101...
  ..11101..
   .1101.
    101
    
    (a) 
\end{verbatim}
\end{minipage}
\begin{minipage}{3cm}
\begin{verbatim}
 ..1011101...
  ..11101..
   .1101.
    101
    
    (b)
\end{verbatim}
\end{minipage}
\begin{minipage}{3cm}
\begin{verbatim}
 ..1101101...
  .101101..
   .1101.
    101
    
    (c)
\end{verbatim}
\end{minipage}
\begin{minipage}{3cm}
\begin{verbatim}
 101010101...
  1010101..
   10101.
    101
    
    (d)
\end{verbatim}
\end{minipage}
\end{center}

\caption{Examples of preimages of 101. Top line
in each diagram represents string of length $3\cdot 3+3=12$, followed by its three consecutive images under  $\mathbf{f}$. Irrelevant symbols are represented by dots.}\label{stringiterates} 
\end{figure}

\begin{proposition} \label{prop010}
The set of $n$-step preimages of 010 under the rule 56528 is the union of the  three sets,
$\mathbf{f}^{-n}(010)=A \cup B \cup C$, defined as follows.
\begin{enumerate}
\item[(i)] $A$ is the set of strings in the form $\underbrace{*...*}_{n-1}00100\underbrace{*...*}_{2n-1}$.
\item[(ii)] $B=\displaystyle\bigcup_{i=0}^{n-1}B_i$, where $B_i$ are the sets of all binary strings of length $3+3n$ of the form
$$\underbrace{*...*}_{i+1}1a_{i+3}...a_{2n-1}10101\underbrace{*...*}_{n-1},$$ and where $a_{i+3}...a_{2n-1}$ contains $n-1-i$ zeros and does not include any 00.
\item[(iii)]  $\displaystyle C= \bigcup_{i=0}^{n-1}C_i$, where $C_i$ are the sets of all binary strings of length $3+3n$ of the form
$$\underbrace{*...*}_{i}001a_{i+4}...a_{2n}1101\underbrace{*...*}_{n-1},$$ and where $a_{i+4}...a_{2n}$ contains $n-1-i$ zeros and does not include any 00.
\end{enumerate}
\end{proposition}
\begin{proof}
By direct listing of all binary strings of length 6 and checking which of them are preimages of 010  we find that
$$\mathbf{f}^{-1}(010)=\left\{00100*\right\} \cup \left\{*10101\right\} \cup \left\{001101\right\} ,$$
therefore
$$\mathbf{f}^{-n}(010)=
\mathbf{f}^{-n-1} (\left\{00100*\right\}) \cup 
\mathbf{f}^{-n-1} (\left\{*10101\right\}) \cup 
\mathbf{f}^{-n-1} (001101) .$$
We will demonstrate that the three sets on the right hand side of the above
correspond to sets $A$, $B$, and $C$.

For the first one, by Proposition 1, the $(n-1)$-step preimages of 00100 have the form $\underbrace{*...*}_{n-1}00100\underbrace{*...*}_{2n-2}$, thus we get  the set $A$ as defined in (i).

For the second set, note that 10101 contains two substrings 101. Preimages of 10101 can be constructed
similarly as preimages of 101 in the proof of Proposition~\ref{prop101}, thus we will not repeat it here. This leads to
the set $B$ of preimage strings as described in (ii).

What remains is to  show that $C=\mathbf{f}^{-n-1} (001101)$, thus we need to construct all $n-1$-step preimages of 001101. Let us first take a look at 
Figure~\ref{pat1} again. 
We can see that blocks 1101 move to the left one cell per time step,
similarly as block 101. Moreover, recall that the left boundary of cluster of zeros
moves to the right upon absorbing 101 arriving from the right. Therefore, every preimage of 
001101 must have the form
$$*\ldots*00a_ja_{j+1}\ldots a_{j+m}1101*\ldots*,$$
where $a_ja_{j+1}\ldots a_{j+m}$ is a buffer
(with the values of $j$ and $m$
are to be determined) which ensures that 1101 is not prematurely destroyed before it arrives to its final position after $n-1$ iterations.  At the same time, this buffer  must contain enough of isolated zeros to allow the 00 on the right to grow
by just enough units so that after $n-1$ iterations the cluster of zeros
ends just before 1101 block, forming the 
desired string 001101.

Suppose now that 
the block 00 is located in the initial string at positions $i$ and $i+1$. The rightmost 0 in the cluster of zeros is therefore at position $i+1$.
After $n-1$ iterations it needs to move to position $(n-1)+2 = n+1$,
because each iteration of $\mathbf{f}$ shortens
the initial string by one cell from the left.
The boundary of the cluster of zeros must therefore
move by $(final position - initial position)=n+1-(i+2)=n-i-1$. This can happen if
the cluster of zeros absorbs exactly
$n-i-1$ zeros, so the buffer must contain exactly
$n-i-1$ zeros. As it turn out,  $n-i-1$ zeros in the buffer is also exactly the right number of zeros 
needed for the
the block $1101$ to move undisturbed
to its final position, one step to the left at each iteration (see proof of Proposition~\ref{prop101} for explanation 
why this happens).

The above leads to the conclusion
that elements of the set $C$ must have the form
$$\underbrace{*...*}_{i}001a_{i+4}...a_{2n}1101\underbrace{*...*}_{n-1},$$ 
where $a_{i+4}...a_{2n}$ contains exactly $n-1-i$ zeros and does not include 00.
The index $i$ can vary from 0 to $n-1$, thus
we obtain $\displaystyle C= \bigcup_{i=0}^{n-1}C_i$ with $C_i$ defined as in (iii).
\end{proof}

\section{Exact results: Block probabilities}
Using the results of the previous section, we can now compute the relevant block probabilities  using eq.~(\ref{bernoulliP}).
Note that the right hand side of eq.~(\ref{bernoulliP}) is a polynomial in two variables $p,q$, and we will
call it \emph{density polynomial}. We will often write the density polynomial using only
one variable $p$, by substituting  $q=1-p$. The quantity $p$ will be called the density, as it represents
the ``density'' of 1s.

Since our CA rule is number-conserving, density polynomials for 0 and 1 are obvious,
\begin{equation}\label{prob0}
 P_n(1)=p, {\,\,\,\,\,\,}  P_n(0)=1-p.
\end{equation}

The density polynomial for $100$ is easy to obtain from From Proposition \ref{prop100} and eq.~(\ref{bernoulliP}). We have
\begin{equation} \label{prob100}
P_n(100)=(p+q)^{3n}pq^2=p (1-p)^2,
\end{equation}
where, as mentioned, we use  $q=1-p$. 
Similarly, from the same Proposition 1, we have
\begin{equation} \label{prob00100}
P_n(00100)=(p+q)^{3n}pq^4=p (1-p)^4.
\end{equation}

The density polynomial for  $101$ is a bit more complicated. The following lemma will be useful.
\begin{lemma}
 The number of strings $b_1b_2\ldots b_m$ which include exactly $k$ ones and do not include any pair 00
is $\displaystyle \binom{k+1}{m-k}$.
\end{lemma}
\begin{proof}
Note that the number of strings $b_1b_2\ldots b_m$ is the same as the number of strings with 1 added before every one of them, i.e., the string $1b_1\dots b_m$.

Since $b_1\dots b_m$ has no pair 00, then the string $1b_1\dots b_m$ can be viewed as a combination of blocks of 10 and blocks of 1. So the number of strings is the same as the number of such combinations.

The length of the combination is the number of ones in the block $1b_1\dots b_m$ which is $k+1$, and the number of block 10 in the combination is the same as the number of zeros in the block $1b_1\dots b_m$ which is $m-k$, giving the number of combinations $\displaystyle\binom{k+1}{m-k}$. 
\end{proof}

We will now apply the above lemma to construct the density polynomial for the block $101$.  In the statement of
Proposition \ref{prop101}, 
the block $a_{i+2}\ldots a_{2n}$ has  $2n-i-1$ symbols, including $n-i$ zeros and
 $2n-i-1 - (n-i) = n-1$  ones. Such string, according to the above lemma,  is realizable in 
$$ \binom{n-1+1}{2n-i-1-(n-1)}=\binom{n}{n-i}$$
possible ways. The density polynomial corresponding to such a block is, therefore,
$$\binom{n}{n-i} p^{n-1} q^{n-i}.$$
This needs to be multiplied by $(p+q)^i p^3q (p+q)^n$, corresponding to the required  prefix and postfix in the
block $\underbrace{*...*}_{i}1a_{i+2}...a_{2n}101\underbrace{*...*}_{n}$,
and then summed over $i$ from $i=0$ to $i=n$.
In the end, we obtain
$$P_n(101)=\sum_{i=0}^n (p+q)^i p^3 q (p+q)^n \binom{n}{n-i} p^{n-1} q^{n-i}
=\sum_{i=0}^n (p+q)^{n+i} \binom{n}{n-i} p^{n+2} q^{n-i+1},$$
which, after carrying out the summation and simplifying, yields
\begin{equation} 
P_n(101)= q p^2 (p+2q)^n(p+q)^n p^n.
\end{equation}
Substituting  $q=1-p$, we obtain
\begin{equation}\label{prob101}
 P_n(101)=(1-p)  (2-p)^n p^{n+2}.
\end{equation}

Having $P_n(101)$ and $P_n(101)$ (eqs. (\ref{prob100}) and (\ref{prob101}), respectively) we can now  compute $P_n(00)$ and $P_n(000)$
using eqs. (\ref{dependentProbs}),
\begin{align} \label{dependentP}
P_n(100)&=P_n(00)-P_n(000), \nonumber \\
P_n(101)&=P_n(0)-2\, P_n(00)+P_n(000). 
\end{align}
Solving the above for $P_n(00)$ and $P_n(000)$ we obtain,
\begin{align}
 P_n(00) &= P_n(0) -P_n(100)-P_n(101),\nonumber \\
 P_n(000) &= P_n(0) -P_n(101)-2 P_n(100),
\end{align}
 and, by substituting density polynomials of eqs. (\ref{prob101}) and (\ref{prob100}),  we finally get
\begin{align} \label{prob00-000}
 P_n(00) &= 1-p -p (1-p)^2 -(1-p)  (2-p)^n p^{n+2}, \nonumber \\
 P_n(000) &= 1-p -(1-p)  (2-p)^n p^{n+2} -2 p (1-p)^2.
\end{align}

A very similar reasoning can be applied to the density polynomial of $010$, using Proposition \ref{prop010}.
Without supplying all details, we just show the calculations,
which are rather self-explanatory.
\begin{align*}
P_n(010) &= pq^4(p+q)^{3n-2}+\sum_{i=0}^{n-1} (p+q)^{i+n} p^{2+n} q^{1+n-i} \binom{n-1}{n-1-i}\\
&+\sum_{i=0}^{n-1}(p+q)^{i+n-1}p^{2+n}q^{2+n-i}\binom{n-1}{n-1-i}\\
&= pq^4(p+q)^{3n-2}+\frac{p^2q^2(p(p+q)(p+2q))^n}{p+q}
\end{align*}
By substituting  $q=1-p$, we finally obtain
\begin{equation}\label{prob010}
P_n(010)=p(1-p)^4+p^2(1-p)^2(p(2-p))^n.
\end{equation}

Equations (\ref{prob0}), (\ref{prob00-000}), and (\ref{prob010})  provide expressions for block probabilities
$P_n(0)$,  $P_n(00)$ $P_n(000)$ and $P_n(010)$. Let us summarize them here:
\begin{align} \label{exactsol}
 P_n(0)&=1-p , \nonumber \\
 P_n(00) &= 1-p -p (1-p)^2 -(1-p)  (2-p)^n p^{n+2}, \nonumber \\
 P_n(000) &= 1-p -(1-p)  (2-p)^n p^{n+2} -2 p (1-p)^2, \nonumber \\
 P_n(010)&=p(1-p)^4+p^2(1-p)^2(p(2-p))^n.
\end{align}
If we take take the limit of $n \to \infty$ in the above, we obtain the ``steady state'' values,
\begin{align} \label{exactsolasy}
 P_\infty(0)&=1-p , \nonumber \\
 P_\infty(00) &= 1-p -p (1-p)^2,  \nonumber \\
 P_\infty(000) &= 1-p -2 p (1-p)^2, \nonumber \\
 P_\infty(010)&=p(1-p)^4.
\end{align}

\section{Local structure approximation}
We will now construct recurrence relations which block probabilities must satisfy. Since
$P_n(0)$,  $P_n(00)$ $P_n(000)$ and $P_n(010)$ can be used to express all remaining block of length up to 3, and
since $P_n(0)$ remains constant, we need to consider only blocks $00$, $000$, and $010$.
Preimages of these blocks,
obtained by direct computation, 
 are
\begin{align*}
\f^{-n}(0)&=\{
0000,
0001,
0010,
0011,
0101,
1000,
1001,
1101\}, \\
\f^{-n}(00)&=\{
00000,
00001,
00010,
00011,
00101,
10000,
10001,
10010,
10011\}, \\
\f^{-n}(000)&=\{
000000,
000001,
000010,
000011,
000101,
100000, \\
&100001,
100010,
100011,
100101\}, \\
\f^{-n}(010)&=\{
001000,
001001,
001101,
010101,
110101\}.
\end{align*}
The above immediately yields the desired recurrence relations, 
\begin{align*}
P_{n+1}(0)&=
P_{n}(0000)+
P_{n}(0001)+
P_{n}(0010)+
P_{n}(0011)+
P_{n}(0101)+
P_{n}(1000)+
P_{n}(1001)+
P_{n}(1101)
 \nonumber\\      
P_{n+1}(00)&=
P_{n}(00000)+
P_{n}(00001)+
P_{n}(00010)+
P_{n}(00011)+
P_{n}(00101)+
P_{n}(10000)
\nonumber\\ 
&+P_{n}(10001)+
P_{n}(10010)+
P_{n}(10011),\nonumber \\ 
P_{n+1}(000)&=
P_{n}(000000)+
P_{n}(000001)+
P_{n}(000010)+
P_{n}(000011)+
P_{n}(000101)+
\nonumber \\ 
&+P_{n}(100000)+P_{n}(100001)+
P_{n}(100010)+
P_{n}(100011)+
P_{n}(100101), \nonumber \\ 
P_{n+1}(010)&=
P_{n}(001000)+
P_{n}(001001)+
P_{n}(001101)+
P_{n}(010101)+
P_{n}(110101).
\end{align*}
Note that in the above we have blocks of length 6 on the right hand side. Similarly as in eq. (\ref{dependentProbs}),
we can express all block probabilities of length of up to 6 by only 32 independent probabilities,  
\begin{gather*}
\{P_n(0), P_n(00), P_n(000), P_n(010), P_n(0000), P_n(0010), P_n(0100), P_n(0110),\\ 
P_n(00000), P_n(00010), P_n(00100), P_n(00110), P_n(01000), P_n(01010), P_n(01100), P_n(01110),\\ 
P_n(000000), P_n(000010), P_n(000100), P_n(000110), P_n(001000), P_n(001010), P_n(001100), P_n(001110),\\ 
P_n(010000), P_n(010010), P_n(010100), P_n(010110), P_n(011000), P_n(011010), P_n(011100), P_n(011110)\}
\end{gather*}
Equations similar to eq. (\ref{dependentProbs}) are then obtained, although because of their length we
omit them here. Using these equations, the recurrence equations for $P_n(00)$ $P_n(000)$ and $P_n(010)$ become
\begin{align} \label{mastereq}
\underline{P_{n+1}(0)}&=\underline{P_n(0)}, \nonumber \\
P_{n+1}(00)&=P_n(00)-\underline{P_n(00100)}+\dashuline{P_n(0010)},\nonumber  \\
P_{n+1}(000)&=P_n(000)-\underline{P_n(00100)}+\dashuline{P_n(0010)},\nonumber \\
P_{n+1}(010)&=2\underline{P_n(00100)}-\dashuline{P_n(0100)}+P_n(010)-\dashuline{P_n(0010)}+\dashuline{P_n(00110)}-\dashuline{P_n(001100)}.
\end{align} 
In the above,  in addition to variables $P_n(0)$, $P_n(00)$, $P_n(000)$ and $P_n(010)$, we have
probabilities which are constant (underlined, by the virtue of eqs. (\ref{prob0}) and (\ref{prob00100})) as well as 
probabilities of longer blocks which do not appear on the left hand side (dashed underline).
The block probabilities which are underlined can obviously be replaced by their respective constant values, while
the others (underlined by the dashed line) can be approximated by probabilities of shorter blocks, 
using the so-called Bayesian extension \cite{gutowitz87a,paper50},
\begin{equation}\label{bayes}
  P(b_1b_2\ldots b_{k+2}) \approx \begin{cases}
                         P(b_1)P(b_2)P(b_3)    & \text{if $k=1$},\\[0.4em]
                      \displaystyle     \frac{P(b_1 \ldots b_{k}) P(b_2 \ldots b_{k+1}) P(b_3 \ldots b_{k+2}) }{P(b_2 \ldots b_{k}) P(b_3\ldots b_{k+1}) }
  & \text{if $k>1$},
                            \end{cases}
\end{equation}
where we assume that the denominator is positive. If the denominator is zero, then
we take  $ P(b_1b_2\ldots b_{k+2}) =0$. We thus obtain
\begin{align*}
 P_n(0010) &\approx \frac{P_n(001)P_n(010) }{P_n(01)} = \frac{(P_n(00)-P_n(000))P_n(010) }{P_n(0)-P_n(00)} ,\\
 P_n(0100) &\approx \frac{P_n(010)P_n(100)}{P_n(10)} =\frac{P_n(010)(P_n(00)-P_n(000))}{P_n(0)-P_n(00)}, \\
 P_n(00110) &\approx \frac{P_n(001) P_n(011) P_n(110)}{P_n(01) P_n(11)} =
                     \frac{(P_n(00)-P_n(000)) (P_n(0)-P_n(00)-P_n(010))^2}{(P_n(0)-P_n(00)) (1-2\, P_n(0)+P_n(00))} ,   \\
 P_n(001100) &\approx \frac{P_n(001) P_n(011) P_n(110) P_n(100)}{P_n(01) P_n(11) P_n(10)} =
                         \frac{(P_n(00)-P_n(000))^2 (P_n(0)-P_n(00)-P_n(010))^2}{(P_n(0)-P_n(00))^2 (1-2\, P_n(0)+P_n(00))},
\end{align*}
where we used eqs. (\ref{dependentProbs}).
After using the above approximations, and changing variables to  $P_n(0)=1-p$, $P_n(00)=x_n$, $P_n(000)=y_n$, $P_n(010)=z_n$,
the equations (\ref{mastereq}) become
\begin{align} \label{xyzeq}
x_{n+1}&=x_n-p(1-p)^4+\frac{(x_n-y_n)z_n}{1-x_n-p} \nonumber \\
y_{n+1}&=y_n-p(1-p)^4+\frac{(x_n-y_n)z_n}{1-x_n-p} \nonumber \\
z_{n+1}&=2p(1-p)^4-\frac{2(x_{n}-y_{n})z_{n}}{1-x_{n}-p}+z_{n}+\frac{(x_{n}-y_{n})(1-p-x_{n}-z_{n})^2}{(1-x_{n}-p)(x_{n}-1+2p)} \nonumber \\
&-\frac{(x_{n}-y_{n})^2(1-p-x_{n}-z_{n})^2}{(1-x_{n}-p)^2(x_{n}-1+2p)}.
\end{align}
We will call these \emph{local structure equations}. Note that we omitted the first equation of eqs. (\ref{mastereq}),
as it merely reflects the fact that $P_n(0)$ is constant.

The claim of local structure theory  is that eqs. (\ref{xyzeq}) well approximate the behaviour of the actual block probabilities, that is, when equations (\ref{xyzeq}) are iterated, the resulting values of $x_n$, $y_n$ and $z_n$
approximate values of $P_n(00)$, $P_n(000)$, and $P_n(010)$ given by eqs. (\ref{exactsol}). 

In order to verify this claim, we iterated eqs. (\ref{xyzeq}) numerically, starting with initial conditions
$x_0=(1-p)^2$, $y_0=(1-p)^3$, and $z_0=p(1-p)^2$. We continued iterations
until conditions $|x_n-x_{n-1}|<10^{-10}$, $|y_n-y_{n-1}|<10^{-10}$ and $|z_n-z_{n-1}|<10^{-10}$ were simultaneously met.
The resulting values, to be called $x_\infty$, $y_\infty$ and $z_\infty$, were recorded, and this was repeated for
100 equally-spaced values of $p\in [0,1]$. 

\begin{figure} 
(a)\includegraphics[width=11cm]{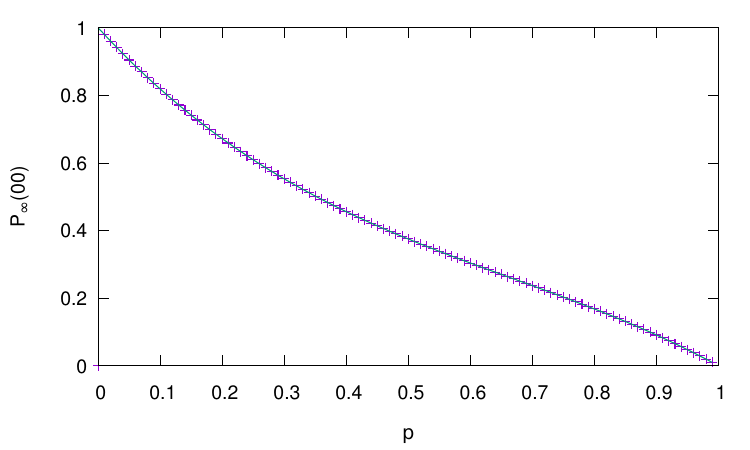}
(b)\includegraphics[width=11cm]{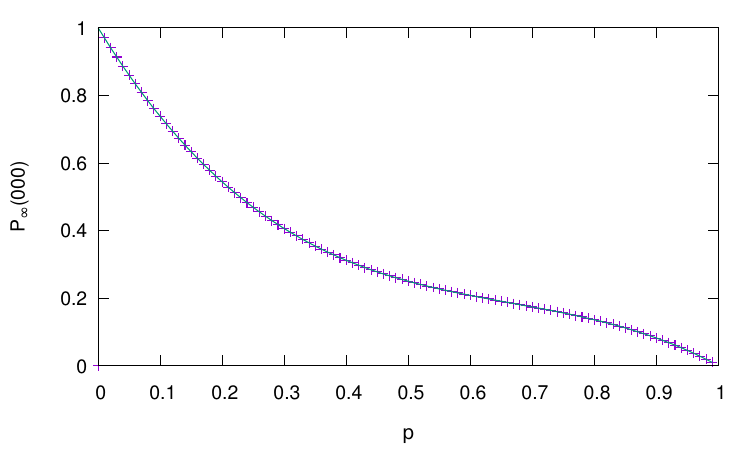}
(c)\includegraphics[width=11cm]{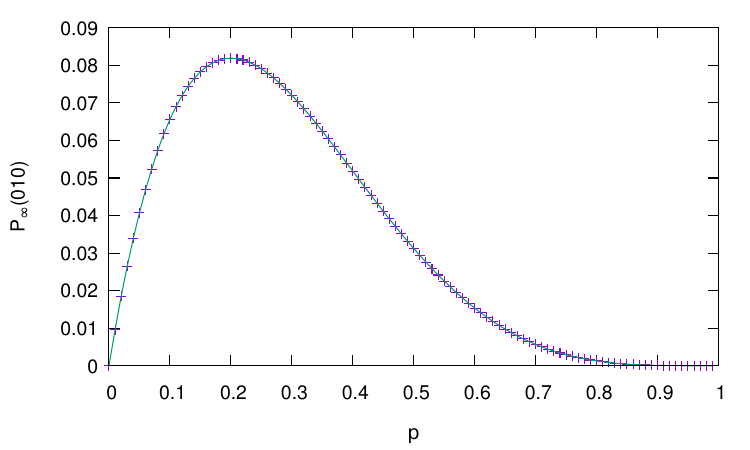}
\caption{Comparison of the exact values of $P_{\infty}(00)$,
$P_{\infty}(000)$ and $P_{\infty}(010)$ (denoted by ``$+$'') with values of $x_\infty$, $y_\infty$, $z_\infty$
obtained by iterating eqs. (\ref{xyzeq}) numerically (denoted by continuous line). 
}\label{fig1} 
\end{figure}
Figure \ref{fig1} show comparison of values of $x_\infty$, $y_\infty$ and $z_\infty$ obtained this way  with exact
values of $P_{\infty}(00)$, $P_{\infty}(000)$ and $P_{\infty}(010)$, as given by eqs. (\ref{exactsolasy}). As we can see, the agreement is excellent, suggesting that
the local structure correctly predicts the asymptotic ($n \to \infty$) values of block probabilities for blocks of length 
up to 3. We will prove the following result.
\begin{proposition}
Let $p \in [0,1]$.  When eqs. (\ref{xyzeq}) are iterated starting with initial conditions
$x_0=(1-p)^2$, $y_0=(1-p)^3$, and $z_0=p(1-p)^2$, then
$(x_n, y_n, z_n) \to (x^\star, y^\star, z^\star)$, where
\begin{align} 
 x^\star &= 1-p -p (1-p)^2,  \nonumber \\
 y^\star &= 1-p -2 p (1-p)^2, \nonumber \\
 z^\star&=p(1-p)^4.
\end{align}
\end{proposition}
\emph{Proof.} 
Let us first note that $x_{n+1}-y_{n+1}=x_n-y_n$, thus $x_{n}-y_{n}=\mathrm{const}$. The value of this constant can be 
determined using initial conditions $n=0$, yielding $x_n-y_n=p(1-p)^2$, and therefore $y_n=x_n-p(1-p)^2$.
Using this, we can reduce eqs. (\ref{xyzeq}) to two components only,
\begin{align} \label{xyzeq-reducd}
x_{n+1}&=x_n-p(1-p)^4+\frac{p(1-p)^2 z_n}{1-x_n-p} \nonumber \\
z_{n+1}&=2p(1-p)^4-\frac{2 p(1-p)^2 z_{n}}{1-x_{n}-p}+z_{n}+\frac{p(1-p)^2(1-p-x_{n}-z_{n})^2}{(1-x_{n}-p)(x_{n}-1+2p)} \nonumber \\
&-\frac{p^2(1-p)^4(1-p-x_{n}-z_{n})^2}{(1-x_{n}-p)^2(x_{n}-1+2p)}.
\end{align}
It is easy to verify that $(x^\star, z^\star)$ is the only fixed point of the above system of difference equations. 
Furthermore, the Jacobian of the mapping defined by eqs. (\ref{xyzeq-reducd}), evaluated at $(x^\star, z^\star)$, is equal to
\begin{equation}
 J= \left[ \begin {array}{lr} {p}^{2}-2\,p+2&\,\,\,\,1\\ \noalign{\medskip}
{p}^{4}-4\,{p}^{3}+3\,{p}^{2}+2\,p-2&\,\,\,\,-1\end {array}
 \right]. 
\end{equation}
One can easily check that its eigenvalues are 
$$
\lambda_{1,2}=\frac{1}{2}+\frac{1}{2}{p}^{2}-p \pm \frac{1}{2}\sqrt {1-4\,p+22\,{p}^{2}-20\,{
p}^{3}+5\,{p}^{4}},
$$
and that $|\lambda_{1,2}|<1$ if $p \in (0,1)$. This means that the fixed point $(x^\star, z^\star)$ is locally stable,
so that $(x_n, z_n) \to (x^\star, z^\star)$ as $n\to \infty$ if the initial point
is sufficiently close to $(x^\star, z^\star)$.
\begin{figure} 
\begin{center}
\includegraphics[width=12cm]{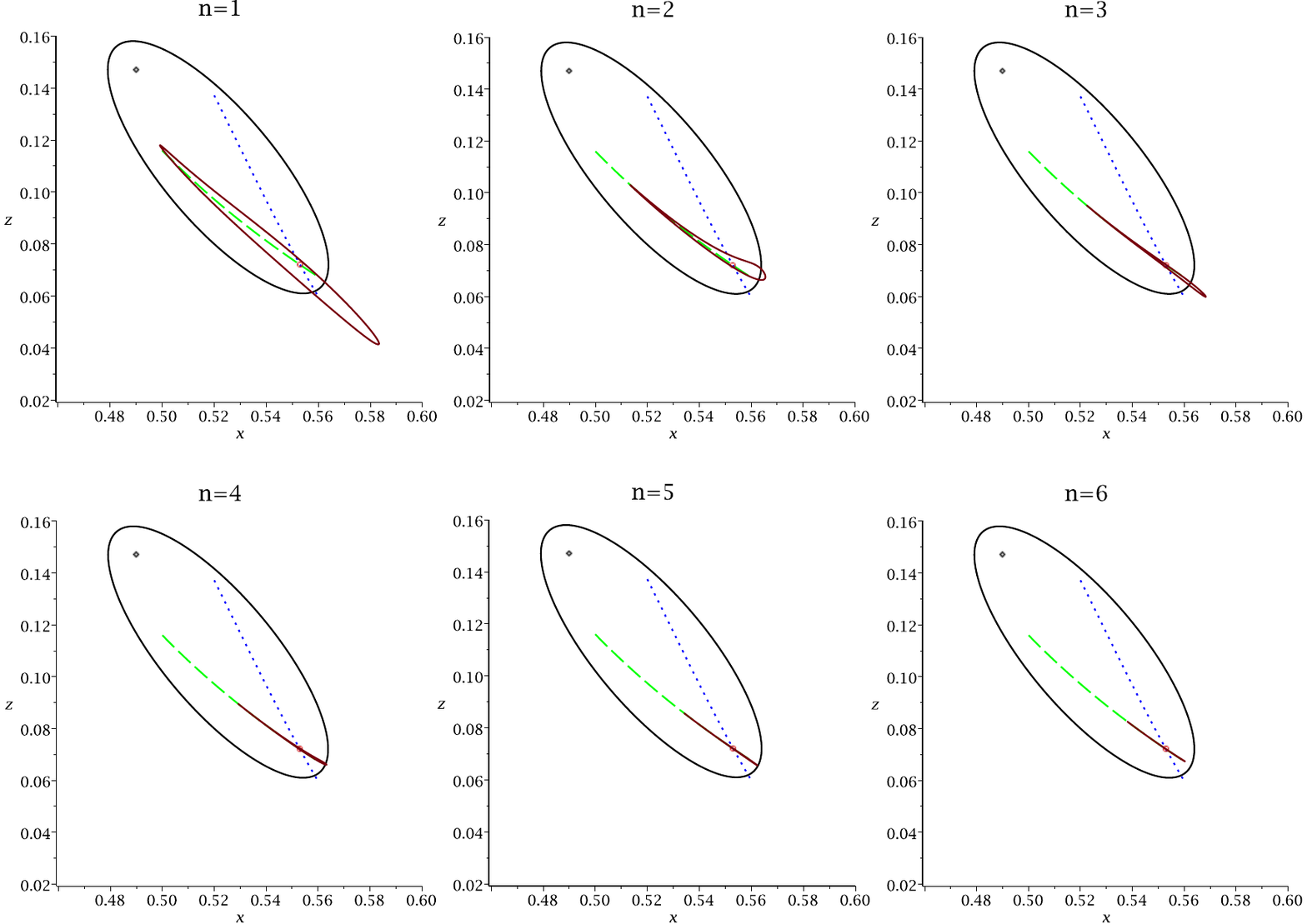}
\end{center}
\caption{Consecutive images (dark red) of a neighbourhoood of the fixed point (black ellipse)  under the map defined by
eqs. (\ref{xyzeq-reducd}), for $p=0.3$. The initial
point $(x_0, z_0)$ is shown as small black diamond. The fixed point $(x^\star, z^\star)$ is represented by small red circle, located at the intersection of two invariant manifolds, shown as dotted and dashed lines (blue and green).
}\label{fig2} 
\end{figure}
Proving formally that $(x_0, z_0) =((1-p)^2, p(1-p)^2)$ is sufficiently close, that is,  lies in the basin of attraction of the stable fixed point
$(x^\star, z^\star)$, is rather difficult, so we will only illustrate it numerically.

Figure~\ref{fig2} shows consecutive images of an elliptical neighbourhood of the fixed point containing $(x_0, z_0)$ under the map defined by
eqs. (\ref{xyzeq-reducd}), using $p=0.3$ as an example. Once can see that the images (red) of the initial disk  shrink with increasing $n$, and that images of all points of the disk converge toward $(x^\star, z^\star)$ (small red circle). The fixed point $(x^\star, z^\star)$ is located
at the intersection of two invariant manifolds, shown as dotted and dashed lines (blue and green). These two manifold
were obtained nymerically by standard methods, similarly as described in Ref. \citen{paper64}.
The green manifold corresponds to  the eigenvalue $\lambda_1$, which is closer to 1 than $\lambda_2$, thus one could call
it ``slow''. Convergence along this manifold is slower than along the blue manifold. Indeed, one can see
that after 6 iterations all points of the image are located almost on the green manifold, or very close to it. Further iterations (not shown) would shrink the image even further along the green line, eventually
converging to a single point. This illustrates that indeed $(x_n, z_n) \to (x^\star, z^\star)$ as $n\to \infty$ if $x_0=(1-p)^2$, $z_0=p(1-p)^2$.

The third component $y_n$  also behaves as expected, 
$y_n \to x^\star - p(1-p)^2= y^\star$, and consequently $(x_n, y_n, z_n) \to (x^\star, y^\star, z^\star)$
for all $p \in (0,1)$. For $p=0$ and $p=1$, direct verification confirms that
$(x_n, y_n, z_n) \to (x^\star, y^\star, z^\star)$ remains valid.

\section{Conclusions}
We presented an example of a rule for which 3-rd order local structure approximation correctly
predicts steady-state probabilities of blocks of length up to 3. This rule possesses additive invariants of order 
1 and three (conserving number of 0s and number of blocks 101).

It is important to note that in order to obtain the aforementioned exact agreement between the local structure theory
and actual values of block probabilities, one needs to make sure that block probabilities which are constant
are included in the local structure equations as constants -- that is, they are not replaced by Bayesian extension
approximations. We experimented with a variant of local structure theory equations in which $P_n(00100)$ is replaced
by its Bayesian extension approximation, and we found that such a variant does not produce correct values
of steady-state block probabilities.

Local structure approximation can be viewed as finite-dimensional (with finite number of block probabilities) approximation of infinitely-dimensional 
dynamical system (there are infinitely many block probabilities needed to describe a measure on $\{0,1\}^\ZZ$).
The existence of a first-order invariant (conservation of the number of zeros and ones) already reduces the
dynamics to a subset of $\{0,1\}^\ZZ$, while
the existence of a blocking word 100  somewhat prevents strong long-range correlations to develop, making 
the system ``nearly finite dimensional''. It seems  to
be reasonable to conjecture that other cellular automata with 
invariants and blocking words should exhibit similar behaviour. In order to explore this further, one needs to find
more examples of ``solvable'' rules first, and this is a problem which we are currently investigating.

\end{document}